\documentclass[11pt]{article}
\usepackage{hyperref}           
\usepackage{amsmath}            
\usepackage{amsfonts}           
\usepackage{graphicx}           
\usepackage{amsthm}             
\usepackage{natbib}             
\usepackage{algorithmic}        
\usepackage[linesnumbered,ruled]{algorithm2e}
\usepackage{xcolor}
\usepackage{url}

\title{Efficient Almost-Egalitarian Allocation of Goods and Bads}
\author{Israel Jacobovich, Ariel University, israel.jacobowi@gmail.com
\\*
Erel Segal-Halevi, Ariel University,  erelsgl@gmail.com}

\date{June 2023}

\newcommand{\er}[1]{\textcolor{blue}{#1}}
\newcommand{\erel}[1]{\er{(Erel says: #1)}}

\newtheorem{theorem}{Theorem}
\newtheorem{lemma}{Lemma}
\newtheorem{corollary}{Corollary}

\begin{document}

\maketitle

\begin{abstract}
We consider the allocation of indivisible 
objects among agents with different valuations, which can be positive or negative.
An *egalitarian* allocation is an allocation that maximizes the smallest value given to an agent; finding such an allocation is NP-hard. 
We present a simple polynomial-time algorithm that finds an allocation that is Pareto-efficient and *almost-egalitarian*: each agent's value is at least his value in an egalitarian allocation, minus the absolute value of a single object.
The main tool is an algorithm for rounding a fractional allocation to a discrete allocation, by which each agent loses at most one good or gains at most one chore. Our algorithm generalizes and simplifies three previous algorithms.

We discuss several aspects and observations about the algorithm and the problem at hand that open doors for efficient and robust implementations.
\end{abstract}
\newpage

\section{Introduction}
Our goal is to allocate some objects fairly among some agents. The set of objects may contain both \emph{goods} (objects with a positive value), \emph{neutrals} (objects with no value) and \emph{bads} (objects with a negative value; also called \emph{chores}). Different agents may assign different values to the same objects; they may even disagree on whether an object is a good, neutral or a chore. For example, patrolling around the camp at night is usually seen as a chore, but may be considered a good for a soldier who likes to see the sunrise.

We assume that the agents' valuations are additive: when they get a set of objects, their total value is the sum of values they assign to the objects.
In order to be as fair as possible, we would like to allocate the objects such that the smallest total value given to a single agent is as large as possible; such an allocation is called \emph{egalitarian}.

If the objects are \emph{divisible}, then we may find a \emph{fractional allocation} --- an allocation in which fractions of the same object may be assigned to different agents. A fractional egalitarian allocation can easily be found by solving a linear program. 
But when the objects are \emph{indivisible}, we must find a \emph{discrete allocation} --- an allocation in which each object must be given entirely to a single agent. Finding a discrete egalitarian allocation is NP-hard even for two agents, by reduction from the \textsc{Partition} problem. 
Therefore, we aim to find an allocation that is \emph{egalitarian up to one object}. This means that, for each agent $i$, if a single chore is removed from $i$'s bundle, or a single good is added to $i$'s bundle, then $i$'s value will be at least as high as his value in some egalitarian allocation.
In other words, $i$'s value is almost as high as the egalitarian value --- the difference is at most the value of a single object.
In fact, our fairness guarantee is somewhat stronger: after removing(adding) a single chore(good), $i$'s value will be at least as high as his value in some \emph{fractional} egalitarian allocation. 
Moreover, the allocation is \emph{fractionally Pareto-optimal (fPO)}, which means that no other allocation, whether discrete or fractional, is at least as good for all agents, and strictly better for some agents. 

Our algorithm operates in three steps, each of which explained in detail below.

\begin{enumerate}
\item  Find a fractional allocation that is egalitarian and fPO (e.g. using linear programming). See section\ref{sub:find-fractional}.

\item Construct the \emph{consumption graph} of the fractional allocation --- a graph in which each object is linked to all agents who receive a positive fraction of it.
Modify the allocation such that its consumption graph becomes a forest (i.e. has no cycles).
See section\ref{sub:make-forest}.

\item Using the cycle-free graph, construct a discrete allocation; the construction ensures that each agent loses at most one good, or gains at most one chore, over his bundle in the fractional allocation. It further guarantees that the allocation remains fPO.
See section\ref{sub:round-allocation}.

\end{enumerate}

Our algorithm extends and simplifies three previously-published algorithms: \cite{lenstra1990} presented a similar algorithm for the problem known as \emph{unrelated machine scheduling}, which is equivalent to almost-egalitarian allocation of only chores; 
\cite{Ivona20005} adapted their algorithm to almost-egalitarian allocation of only goods; 
\cite{Aziz2020} presented an algorithm for almost-proportional allocation of both goods and chores,%
\footnote{
A \emph{proportional allocation} is an allocation in which each agent receives at least $1/n$ of his value for the set of all objects.
A proportional fractional allocation can be trivially computed by giving each agent $1/n$ of every good.
Using this trivial allocation in Step 1 of our algorithm, instead of the egalitarian allocation, yields an allocation that is fPO and almost-proportional (also known as PROP1).
}
but their rounding algorithm is more complex than ours. See Section \ref{sec:related} for a more detailed comparison to previous work.

Our algorithm can be adapted to different fairness criteria by modifying step 1. For example, when allocating only divisible goods, it is quite common to use the fractional allocation that maximizes the product of utilities, also known as the Max Nash Welfare allocation. If  such an allocation is used in step 1, instead of the egalitarian allocation,   then our algorithm yields an allocation in which each agent receives at least his value in some Max Nash Welfare allocation minus the absolute value of a single object.

\section{Preliminaries}
An instance of the fair division problem is represented by a matrix $U \in \mathbb{R}^{n \times m}$, where $n$ is the number of agents and $m$ is the number of objects. $u_{i o}$ refers to agent $i$'s valuation of object $o$, which can be considered as bad ($u_{i o} < 0$), good ($u_{i o} > 0$), or neutral ($u_{i o} = 0$).
We define $A$ to be the set of \emph{agents} and $O$ the set of \emph{objects} (or \emph{items}).

A \emph{discrete allocation} is a matrix $X \in \{0, 1\}^{n \times m}$ such that
$$\forall o \in [m]: ~~~~ \sum_{i=1}^{n} x_{i o} = 1,$$
It describes a solution to $U$ where each object is allocated to exactly one agent.
A \emph{fractional allocation} is a matrix $X \in  [0, 1]^{n \times m}$
with the same constraint. It describes a solution where splitting an object between two or more agents is possible.
Agent $i$'s utility under allocation $X$ is $u_i(X) = \sum_{o=1}^{m} u_{ij} \cdot x_{i o}$.

The \emph{consumption graph} of a fractional allocation $X$, denoted as $G_X$, is a bipartite graph defined as follows:
\begin{align*}
V(G_X) &= A\: \mathaccent\cdot\cup \:O \\
E(G_X) &= \{(i, o) \mid x_{i o} > 0\}
\end{align*}
Note that the set of agents and set of objects are the partition classes.

An allocation $Y$ \emph{weakly Pareto-dominates} an allocation $X$ if $u_i(Y)\geq u_i(X)$ for all agents $i$.
$Y$ \emph{strongly Pareto-dominates} $X$ if, in addition, $u_i(Y)> u_i(X)$ for at least one agent $i$. $Y$ is also said to be a \emph{Pareto improvement} over $X$ in that case.

An allocation $X$ is called \emph{fractionally Pareto-optimal  (fPO)} if no other allocation, whether discrete or fractional, strongly Pareto-dominates it.

\section{The Algorithm}
\subsection{Finding a fractional fair allocation}
\label{sub:find-fractional}
Finding a fractional fair allocation is simple. For example, a fractional egalitarian allocation $X$ can be computed
by the following linear program:
\begin{align*}
\tag{*}
\text{maximize}   && z
\\
\text{subject to} 
&&
z &\leq \sum_{o=1}^{m} x_{i o} u_{i o} && \text{for all $i\in A$}
\\
&& 
\sum_{i=1}^{n} x_{i o} &= 1
&& \text{for all $o\in O$}
\\
&& 
x_{i o} &\geq  0
&& \text{for all $o\in O, i\in A$}
\end{align*}

The allocation $X$ returned from (*) is not necessarily fPO, but it is simple to find an fPO allocation that is a (weak) Pareto-improvement of any given allocation $X$, using the following LP:

\begin{align*}
\tag{**}
\text{maximize}   
&& 
\sum_{i=1}^n
\sum_{o=1}^{m} 
y_{i o} u_{i o}
\\
\text{subject to} 
&&
\sum_{o=1}^{m} y_{i o} u_{i o} &\geq u_i(X) && \text{for all $i\in \{1,\ldots, n\}$}
\\
&& 
\sum_{i=1}^{n} y_{i o} &= 1
&& \text{for all $o\in O$}
\\
&& 
y_{i o} &\geq  0
&& \text{for all $o\in O, i\in A$}
\end{align*}

Note that in this LP, the allocation $X$ and the utilities $u_i(X)$ are constants. 
The first constraint guarantees that $Y$ is a weak Pareto-improvement of $X$, and the other two constraints guarantee that $Y$ is a legal allocation.
The LP is clearly feasible, since $X$ solves it.
The LP objective maximizes the sum of utilities of all agents; this guarantees that $Y$ is fPO.

\subsection{Ensuring that the consumption graph has no cycles}
\label{sub:make-forest}

The next step is to modify the allocation $X$ such that its consumption graph $G_X$ is acyclic. This should be done without harming any agent, that is, the modified allocation should be a weak Pareto-improvement of $X$.

There are two ways to do make $G_X$ acyclic: one way is to iteratively remove edges from $G_X$ until $G_X$ has no cycles. The other way is to use the Simplex method. 

\subsubsection{Removing cycles: iterative procedure}
We describe the iterative procedure first.%
\footnote{
A similar procedure appeared in \citet{bogomolnaia2017competitive}  and \citet{Erel2022}.
Their algorithm is more complex than ours, but it has a stronger guarantee: it guarantees that the final allocation is fPO even if the initial allocation is not.}
Initially, we look for individual edges that can be removed without harming any agent. 
In particular, if some object $o$ is split between two agents $i$ and $j$, such that $u_{i o}\leq 0$ and $u_{j o}\geq 0$, then $i$ can give all his holdings in $o$ to $j$ without harming any of them, thus removing the edge $\{i,o\}$ from $G_X$. 
Note that, if the original allocation is fPO, there is no object $o$ split between agents $i$ and $j$ s.t:  $u_{i o}\leq 0$ and $u_{j o}> 0$ or $u_{i o}< 0$ and $u_{j o}\geq 0$, so the only effect of this step is to remove edges with $u_{i o} =u_{j o} = 0$.

Once these ``easy'' edges are removed, if any object $o$ is still split between $i$ and $j$, then either both of them value it positively, or both of them value it negatively.
Therefore, every change in the allocation of $o$ between $i$ and $j$ harms one of these agents and helps the other one. 
In particular, suppose some fraction $\epsilon$ is moved from $i$ to $j$:
\begin{itemize}
    \item If both agents value $o$ positively, then $i$ loses $\epsilon\cdot u_{i o}$ utility, and $j$ gains  $\epsilon\cdot u_{j o}$ utility.
    \item If both agents value $o$ negatively, then $i$ gains $\epsilon\cdot |u_{i o}|$ utility, and $j$ loses $\epsilon\cdot |u_{j o}|$ utility.
\end{itemize}
If a change in the allocation of $o$ between $i$ and $j$ increases the utility of $j$ by some amount $u>0$, we say that we ``transfer $u$ utility to $j$ using $o$''; if $j$'s utility decreases by $u$, we say that we ``transfer $u$ utility from $j$ using $o$. In particular:
\begin{itemize}
    \item If both agents value $o$ positively, then to transfer $u$ utility to $j$, we should move $u/u_{j o}$ amount of $o$ from $i$ to $j$. In this case, $i$ loses  $u\cdot u_{i o}/u_{j o}$ utility.
    \item If both agents value $o$ negatively, then to transfer $u$ utility to $j$, we should move $u/|u_{j o}|$ amount of $o$ from $j$ to $i$. In this case, $i$ loses  $u\cdot |u_{i o}|/|u_{j o}|$ utility.
\end{itemize}
In both cases, $j$ gains $u$ utility and $i$ loses $u\cdot |u_{i o}|/|u_{j o}|$ utility.

Suppose now that there is a cycle in $G_X$. For simplicity, we consider a cycle with $3$ agents and $3$ objects; the reader will have no difficulty generalizing the idea to cycles of any length. Let us denote the agents by $a,b,c$ and the objects by $1,2,3$, such that the cycle is: $a-1-b-2-c-3-a$. We can transfer utility along the cycle in two directions: from $a$ to $b$ to $c$ to $a$, or from $a$ to $c$ to $b$ to $a$. In particular:
\begin{enumerate}
    \item In one direction, we transfer $u$ utility to $b$ using object 1; then $a$ loses $u\cdot |u_{a 1}|/|u_{b 1}|$ utility. We compensate $a$ by transferring $u\cdot |u_{a 1}|/|u_{b 1}|$ utility to $a$ using object 3; then  $c$ loses $u\cdot (|u_{a 1}|/|u_{b 1}|)\cdot (|u_{c 3}|/|u_{a,3}|)$ utility. We compensate $c$ by transferring the same amount of utility to $c$ using object 2; then $b$ loses $u\cdot \big[(|u_{a 1}|/|u_{b 1}|)\cdot (|u_{c 3}|/|u_{a,3}|)\cdot (|u_{b 2}|/|u_{c 2}|)\big]$ utility; denote the quantity in brackets by $r$. Overall, $a$ and $c$ remain with the same utility, whereas $b$'s utility changes by $u\cdot(1-r)$.
    \item In the opposite direction, we transfer $u$ utility from $b$ using object 1; then $a$ gains $u\cdot |u_{a 1}|/|u_{b 1}|$ utility. Now we transfer $u\cdot |u_{a 1}|/|u_{b 1}|$ utility from $a$ using object $3$; then $c$ gains $u\cdot (|u_{a 1}|/|u_{b 1}|)\cdot (|u_{c 3}|/|u_{a,3}|)$ utility. Finally, we transfer the same amount of utility from $c$ using object 2; then $b$ gains utility $u\cdot [(|u_{a 1}|/|u_{b 1}|)\cdot (|u_{c 3}|/|u_{a,3}|)\cdot (|u_{b 2}|/|u_{c 2}|)] = u\cdot r$. Overall, $a$ and $c$ remain with the same utility, whereas $b$'s utility changes by $u\cdot(r - 1)$.
\end{enumerate}
Clearly, either $u\cdot(1-r)$ or $u\cdot(r-1)$ or both are non-negative, so there is at least one direction in which the change in the utility of $b$ is non-negative. We choose one such direction; every transfer of utility in this direction is a weak Pareto-improvement. 

Now, we compute how much utility to move in the chosen direction. We have to compute $u$ such that, after the cyclic transfer, at least one edge in the cycle disappears. For each of the three transfers, we compute how much fraction of the relevant object should be moved. For example, to move $u$ utility from $a$ to $b$ using object $1$, we have to either 
take $u/|u_{b 1}|$ fraction from $a$ (if $1$ is a good), or take $u/|u_{b 1}|$ fraction from $b$ (if $1$ is a bad).
In the first case, we must have $u \leq x_{a 1}\cdot |u_{b 1}|$; in the second case, we must have $u \leq x_{b 1}\cdot |u_{b 1}|$.
Similarly, for each object in the cycle, we compute an upper bound on $u$. We choose $u$ to be the smallest of these upper bounds. This $u$ satisfies all the required inequalities, and at least one of them is satisfied with equality. Therefore, at least one edge is removed from $G_X$.

A pseudo-code of the cycle-removal process is given in Algorithm \ref{algo:ToForestTrade}.

\begin{algorithm}[H]
\SetAlgoLined
\SetKwInOut{Input}{Input}
\SetKwInOut{Output}{Output}
\Input{
A utility matrix $U$ and a fractional allocation $X$.}
\Output{
A fractional allocation $Y$ such that (1) $Y$ weakly Pareto-dominates $X$, (2) the  graph $G_Y$ is acyclic.}

$Y := X$.

\If{some $o$ is split between $i$ and $j$ s.t: $u_{i o} \geq 0$ and $u_{j o} \leq 0$}{ update $Y$ to allocate all of $j$'s share of $o$ to $i$.}

Let $G_Y$ be the consumption graph of $Y$.

\While{$G_Y$ is cyclic}{

$C \leftarrow \operatorname{FindCycle}(G_Y);$ // \CommentSty{E.g. using the BFS algorithm.}

$\operatorname{RemoveCycle(U, Y, C);}$ // \CommentSty{See description in text.}

}

\caption{ToForestTrade}
\label{algo:ToForestTrade}
\end{algorithm}

The discussion in the text before the algorithm proves the following theorem.
\begin{theorem}
For any input allocation $X$,
Algorithm \ref{algo:ToForestTrade} returns an output allocation $Y$ that weakly Pareto-dominates $X$, such that $G_Y$ has no cycles.
\end{theorem}

Note that, if the input allocation $X$ is fPO, then the output allocation $Y$ is fPO too, since it is a weak Pareto-improvement of $X$.

We now analyze the runtime Algorithm \ref{algo:ToForestTrade}.
Denote  $E_X = |\{(i, o) \mid x_{i o} > 0\}|$, i.e. how dense the original fractional allocation matrix is.
This is exactly the number of edges in the consumption graph $G_X$. Note that $E_X \leq m \cdot n$.
The construction of $G_X$ takes $O(E_X)$ time.

The initial process of removing the ``easy'' edges takes at most 
$O(E_X)$.

Finding a cycle in $G$ can be done in $O(E_X)$ time, e.g. using BFS.
\emph{RemoveCycle} takes $O(\text{length of cycle}) = O(E_X)$ time.
For this reason, we believe BFS may be practically preferred to DFS for finding a cycle, as it finds shorter cycles first, and thus reduces the run-time of \emph{RemoveCycle}.

Since at least one edge is removed at each iteration, the number of iterations is  bounded by the number of edges $= O(E_X)$.
So the worst case run-time is $O(E_X^2)$
.

\subsubsection{Removing cycles: Simplex-based procedure}
There is another way to make $G_X$ acyclic, which is simpler to implement in practice; it uses the linear program (**) from subsection \ref{sub:find-fractional}.

\begin{lemma}
If the linear program (**) is solved using the Simplex method, then the returned allocation $Y$ has an acyclic consumption graph.%
\footnote{
The following lemma is mentioned briefly in \citet{Aziz2020}. We are grateful to Fedor Sandomirskiy for explaining the proof to us.
}
\end{lemma}
\begin{proof}
It is well-known that the Simplex method returns a solution that is a \emph{corner} (an extreme point) of the space of all feasible solutions. In other words, it returns a solution that is not an average of two different solutions. We use this fact to prove that $G_Y$ has no cycles.
 
Suppose by contradiction that $G_Y$ has a cycle; for simplicity, denote the cycle by $a - 1 - b - 2 - c - 3 - a$ (as above, the reader should have no difficulty in generalizing the argument to cycles of any length). As explained in the beginning of this subsection, it is possible to transfer utility along the cycle in two opposite directions, such that agents $a$ and $c$ remain with the same utility. In one direction, the utility of $b$ changes by $(1-r)u$; in the opposite direction, it changes by $(r-1)u$. Let $Y'$ and $Y''$ respectively be the resulting allocations.

If $r<1$, then $Y'$ strongly Pareto-dominates $Y$;
if $r>1$, then $Y''$ strongly Pareto-dominates $Y$; 
both cases contradict the fact that $Y$ is fPO.
Therefore, $r = 1$, so the utilities of \emph{all} agents are the same in all three allocations $Y,Y',Y''$.
This means, in particular, that all three allocations satisfy the constraints of LP (**). But, by construction, $Y$ is an average of $Y'$ and $Y''$; this contradicts the fact that $Y$ is an extreme point  of (**).
\end{proof}

\subsection{Rounding the fractional allocation}
\label{sub:round-allocation}

The final step is to 'round' the fractional allocation $X$ into a discrete allocation, by giving each object to a single agent among those who own parts of it in $X$.
Algorithm \ref{algo:rounding} is our simplified rounding algorithm.

\begin{algorithm}
\SetAlgoLined
\SetKwInOut{Input}{Input}
\SetKwInOut{Output}{Output}
\Input{
A utility matrix $U$
and a fractional allocation $Y$, such that:
(1) $G_Y$ has no cycles,
(2) Every object is split only between agents who view it as good, or only between agents who view it as bad.
}
\Output{
A discrete allocation $Z$.
}

Initialize $Z := Y$ and $G := G_Y$.

\While{$G$ is non-empty}{

Let $l$ be any leaf of $G$ (it must exist since $G$ is acyclic);

\uIf{leaf $l$ is an object $o$}{
    Let $i$ be the only neighbor of $o$ (who must be an agent);
    
    Change $Z$ to allocate $o$ entirely to $i$.
}
\uElseIf{leaf $l$ is an agent $i$
and its neighbor $o$ has a neighbor $k \neq i$ for whom $u_{k, o} > 0$} {
Choose an arbitrary such neighbor $k$;  change $Z$ by moving all the share of agent $i$ in $o$ to $k$.
\label{step:good}
}
\uElse{  
// \CommentSty{\small All neighbors $k\neq i$ of $o$, if any, perceive $o$ as bad.}

Change $Z$ to allocate $o$ entirely to $i$.
\label{step:bad}
}

Remove leaf $l$ from the forest $G$.
\label{step:remove-leaf}
}

\caption{Rounding a Fractional Allocation}
\label{algo:rounding}
\end{algorithm}

We now analyze the properties of the algorithm.

\subsubsection{Economic efficiency}
\begin{theorem}
Algorithm \ref{algo:rounding} preserves fPO, that is: if the input allocation is fPO, then the output allocation is fPO too.
\end{theorem}

\begin{proof}
The proof is very similar to Lemma 2 of \citet{Aziz2020}. It uses the following lemma by \citet{varian1976two}:

\begin{lemma}
\label{lem:fpo1}
A (fractional) allocation $X$ 
is fPO 
if and only if 
it maximizes a weighted sum of utilities:
$\sum_{j=1}^{m} u_i(X) \lambda_i$,
for some strictly-positive vector $\lambda \in \mathbb{R}_+^m$.
\end{lemma}
Suppose the input allocation $Y$ is fPO, and let $\lambda$ be the vector guaranteed by Lemma \ref{lem:fpo1}.

Suppose in $Y$ some object $o$ is split between a subset of agents $S \subseteq A$.
Suppose we move a fraction of $o$ from some agent $a_i\in S$ and another agent $a_j\in S$.
By \ref{lem:fpo1}, this transfer cannot increase the weighted sum 
$\sum_{j=1}^{m} \lambda_i\cdot u_i(Y)$.
Therefore, transferring a fraction of $o$ in the other direction cannot \emph{decrease} the sum.
In other words, every transfer of parts of $o$ among agents of $S$ maintains the sum 
$\sum_{j=1}^{m} \lambda_i\cdot u_i(Y)$.

Since Algorithm \ref{algo:rounding} only reallocates an object $o$ between neighbours of $o$ in $G_Y$, the output allocation $Z$ still maximizes the sum $\sum_{j=1}^{m} \lambda_i\cdot u_i(Z)$, with the same vector $\lambda$. Therefore, by Lemma \ref{lem:fpo1}, $Z$ is fPO.
\end{proof}

\subsubsection{Fairness up to one object}

\begin{theorem}
For any fractional allocation $Y$ input to Algorithm \ref{algo:rounding}, 
the output $Z$ is a valid discrete allocation satisfying, for every agent $i$,
\begin{align*}
    u_i(Z) \geq u_i(Y) - \max_{o \in [m]} |u_{i o}|,
\end{align*}
that is, the utility of $i$ in the returned allocation is at least $i$'s utility in the original allocation minus the absolute value of a single object.
\end{theorem}

\begin{proof}
we first prove by induction that the algorithm satisfies the following invariant:

\begin{quote}
For every agent $i$, in every iteration before the one in which $i$ leaves the graph, $u_i(Z) \geq u_i(Y)$.  
\end{quote}

\emph{Basis}:
At the 0th iteration (before any leaf is chosen), $Z = Y$, so the invariant is met.

\emph{Step:} 
We assume the invariant holds  before iteration $t \geq 1$, and prove that they hold after iteration $t$.

If the leaf chosen at iteration $t$ is some object $o$, then we must have $z_{o j}=1$, where $j$ is the only neighbor of $o$. Therefore, allocating it entirely to $j$ does not change the utility of any agent, so the invariant is still met.

Suppose the leaf chosen at iteration $t$ is some agent $j\neq i$, and let object $o$ be the only neighbor of $j$.
If $i$ does not have any share of $o$, then he will not be affected by the removal of $j$.
If $i$ has a share of $o$ and values it positively, then in step \ref{step:good}, $i$ either gets more shares of $o$, or does not get anything; 
if $i$ values $o$ negatively, then 
in step \ref{step:bad}, he might 'lose' his share in $o$.
In any case, $u_i(Z)\geq u_i(Y)$ still holds.
This concludes the proof of the invariant.



Now, consider the step in which $i$ is removed from the graph. 
In step \ref{step:good}, $i$ loses his share in a good object $o$; in step \ref{step:bad}, $i$ gains shares in a bad object $o$. In both cases, his utility decreases, but the decrease is at most $|u_{i o}|$.
After that, $i$ is removed from the graph and will not lose any more utility.
Therefore, by the invariant, 
$u_i(Z) \geq u_i(Y) - \max_{o \in [m]} |u_{i o}|$.

Each object is removed from $G$ only when it is allocated to a single agent. Therefore, the output $Z$ is a valid discrete allocation.
\end{proof}

\subsubsection{Run time of Algorithm \ref{algo:rounding}}
The construction of $G_Y$ takes $O(E_Y)$ time. Note that, since $Y$ is acyclic, $E_Y$ is at most the number vertices minus 1, that is, $E_Y\leq m+n-1$.

In each iteration one vertex is removed from $G$, so there are at most $m+n$ iterations.

In each iteration, we have to find a leaf. We can do this efficiently by constructing a set $L$ of all leafs of $G_Y$, and updating $L$ in each iteration. Note that, when a leaf is removed, the only possible new leafs are its neighbors. Every node enters $L$ and exits $L$ at most once. Therefore, the total run-time of all iterations is $O(m+n)$.

\subsection{Practical improvements}
We have seen that the running time of the algorithm is dominated by the first step (finding a fractional allocation) or the second step (\emph{ToForestTrade}).



One way to make \emph{ToForestTrade} faster is to apply parallel computing, in case the consumption graph of the intermediate allocation $G_Y$  turns out to be disconnected.
Note that, in this case, each connected component of $G_Y$ contains a complete fractional allocation of a subset of the objects among a subset of the agents. 
Therefore, we may apply Divide \& Conquer by rerunning \emph{ToForestTrade} 
on the allocation of each connected component separately.  
At the end of the recursion depth, the consumption graph is always connected. At each call the input gets smaller too.
This idea can be implemented via truly parallel computation of each call.

There might be other ways to practically improve the run-time, which we leave as open questions:

\begin{itemize}
\item 
Are there any ways to reduce the number of iterations in \emph{ToForestTrade} (Algorithm \ref{algo:ToForestTrade})?
\item
Is there a way to decide efficiently in advance (before running the linear programs of Section \ref{sub:find-fractional}), whether a given instance will be split into two independent instances?
\end{itemize}

\section{Related Work}
\label{sec:related}
\cite{lenstra1990} studied the problem of job scheduling on unrelated parallel machines. In this problem, the input is a list of jobs and a list of machines, such that each job can have a different run-time on each machine. The goal is to minimize the maximum completion time of a machine. It is easy to see that the problem is equivalent to egalitarian allocation of objects with negative utilities. Before their work, most of the results regarding this problem were unpractical or restricted to a weaker version of the problem. They showed a 2-approximation, strongly polynomial-time algorithm.
Their paper introduced the linear programming rounding approach.
They continued by carefully rounding the solution. The tool that enabled this scaffolding is a what is now known as the consumption graph.
The nature of the linear program yielded a consumption graph that is a pseudo-forest - a forest with possibly one additional edge. This lead to an algorithm that is more complex than ours, since they had to handle both leafs and cycles, rather than just leafs.
Their technique was later adopted and used throughout the literature to explore other variants of fair allocation.
In particular,  this technique was used for the problem of max-min (egalitarian) allocation of indivisible \emph{goods}.

\cite{Golovin20005} gave an algorithm that can deliver at least $\frac{1}{k}$ the utility
of the optimal solution for $1 - \frac{1}{k}$ fraction of the agents for any integer $k$, and a $O( \sqrt{n} )$ approximation of the max-min value for a special case with only two classes of goods.
\citet{Ivona20005} demonstrated several algorithms for approximate egalitarian object allocation for positive additive or maximal utilities.
One of them appealed to the rounding approach, in that case the consumption graph turned out to always be a pseudo tree.
They changed the rounding procedure by a bit to get an additive
approximation guarantee against the fractional optimum.

\cite{Arash20007}
gave an approximation algorithm for egalitarian goods allocation, with an approximation factor of $\Omega ( \frac{1}{\sqrt{n} \log^3 n} )$.

Beyond egalitarian fairness:

 \cite{Aziz2020} gave yet another algorithm that implemented the rounding approach.
It was building upon an algorithm by \cite{Erel2022} that gave a fractional solution that is proportional and fPO with a forest consumption graph.
Their algorithm also augmented the rounding procedure in a way that allowed it to deal with mixed utilities and give a fPO and almost-proportional guarantee. Their rounding algorithm is more complex than ours, using two nested loops and a queue, rather than just a single flat loop.

All of the aforementioned contributions solve similar versions of the same problem at different levels of difficulty with no consistency in the rounding stage, opting for different fairness objectives.

Our Algorithm serves as a universal generic method. It can take any input (mixed additive utilities without any extra stipulations).
Its main contribution is that it performs the rounding stage in a simpler way than previous algorithms, without charging any effect on the guarantees mentioned above.

Moreover, thanks to the procedural generic design of the algorithm we were able to point out ways to improve practical implementations.

\newpage
\bibliography{bibliography}

\end{document}